\documentclass[letterpaper, 10 pt, conference]{ieeeconf}

\usepackage{graphicx}
\usepackage{epstopdf}

	\usepackage{generic}
	\usepackage{cite}
	\usepackage{amsmath,amssymb,amsfonts}  
	\usepackage{anyfontsize}
	\def\BibTeX{{\rm B\kern-.05em{\sc i\kern-.025em b}\kern-.08em
			T\kern-.1667em\lower.7ex\hbox{E}\kern-.125emX}}
	\allowdisplaybreaks
	\makeatletter
	\let\NAT@parse\undefined
	\makeatother
	\usepackage[pdftex, pdfborderstyle={/S/U/W 0}]{hyperref}

	\usepackage{graphicx,graphics,epsfig,color,wrapfig}
	\usepackage{float,subcaption}
	\usepackage{pgf,tikz,tikz-3dplot}
	\usepackage{tikzit}
	\usetikzlibrary{arrows,arrows.meta}
	\usetikzlibrary{positioning}
	\usepackage{pgfplots} 
	\usetikzlibrary{external}
	\usepackage{threeparttable}

	\newtheorem{exe}{Example}
	
	\newtheorem{rem}{Remark}
	\newtheorem{thm}{Theorem}
	\newtheorem{lem}{Lemma}
	\newtheorem{prop}{Proposition}
	\newtheorem{corol}{Corollary}
	
	\newtheorem{defin}{Definition}
	\newtheorem{assump}{Assumption}

	\newenvironment{proposition}{\begin{prop}}{\hfill $\square$ \end{prop}}
	
	\newenvironment{remark}{\begin{rem}}{\hfill $\bullet$ \end{rem}}
	
	\newenvironment{lemma}{\begin{lem}}{\hfill $\square$ \end{lem}}
	\newenvironment{definition}{\vskip 3pt \begin{defin}}{\vskip 3pt \end{defin}}
	\newenvironment{assumption}{\vskip 3pt\begin{assump}}{\end{assump}\vskip 3pt}

	\usepackage{lipsum} 
	\newif\ifitsdraft
	
	\newif\ifitsSPL
	
	\newif\ifwithlualatex

		\definecolor{myred}{rgb}{0.86,0.1,0.16}
		\definecolor{gray}{rgb}{0.3,0.3,0.3}\def\gray#1{{\textcolor{gray}{#1}}}
		\definecolor{lightgray}{rgb}{0.6,0.6,0.6}
		\definecolor{gray}{rgb}{0.33,0.4,0.47}\def\gray#1{{\color{gray}#1}}
		\definecolor{steelblue}{rgb}{0,.42,.7}
		\definecolor{britishgreen}{rgb}{0,0.26,0.15}
		\definecolor{navyblue}{rgb}{0,0,.8}
		\definecolor{olivegreen}{rgb}{0.14,0.29,0}
		\definecolor{myred}{rgb}{0.86,0.1,0.16}

	\ifitsdraft 
	
	\usepackage{refcheck}  
	\usepackage{soul}\onecolumn
	\usepackage[colorinlistoftodos,prependcaption,textsize=footnotesize,textwidth=2in]{todonotes}
	
	\onecolumn
	\newcounter{er}    \newcounter{al}    \newcounter{is}

	\newcommand{\rmk}[2]{
		\begin{quote}
			\color{navyblue} {\bf #1:} #2 \color{black}
		\end{quote}
	}

	\newcommand{\marginER}[2][0]%
	{\stepcounter{er}
		\todo[linecolor=red,backgroundcolor=red!1,bordercolor=red,noline]{\color{black}ER(\theer): #2}}

	\newcommand\todoin[1]{\stepcounter{er}
		\todo[inline, caption={}, linecolor=red,backgroundcolor=red!1, bordercolor=red]%
		{\begin{minipage}{\textwidth-4pt}\color{black}ER(\theer): #1\end{minipage}}}

	\newcommand{\inlineER}[1]
	{\stepcounter{er}
		\todo[linecolor=red,backgroundcolor=red!1,bordercolor=red,inline]{
			\color{black}ER(\theer): #1}}

	\newcommand{\draftfootnote}[1]{\footnote{\color{gray}{\hrule \ \\ DRAFT FOOTNOTE: #1}}}

	\else
		\ifitsSPL   \onecolumn \textwidth=6.7in\oddsidemargin=0in\evensidemargin=0in%
		\renewcommand{\baselinestretch}{1.1}
	  \fi
		 
		\newcommand\todoin[1]{}

		\newcommand{\marginER}[2][0]{}
		\newcommand{\inlineER}[1]{}
		
		\newcommand{\rmk}[2]{} 
		\newcommand{\draftfootnote}[1]{}\def\gray#1{}

\fi

\title{{ 
    \bf Simultaneous Topology Estimation and Synchronization of Dynamical Networks with Time-varying Topology
      }}
\author{Nana Wang, Esteban Restrepo, and Dimos V. Dimarogonas
  	\thanks{This work is supported by the Swedish Research Council (VR), the	Knut and Alice Wallenberg Foundation, and the WASP-DDLS program.
  		N. Wang and D. V. Dimarogonas are with the Division of Decision and Control Systems, KTH Royal Institute of Technology, SE-100 44 Stockholm, Sweden, email: \texttt{\{nanaw,dimos\}@kth.se}. E. Restrepo is with CNRS-IRISA, Inria Rennes, France, email: \texttt{esteban.restrepo@inria.fr}.
        }
    }
\pgfplotsset{compat=1.18}
\begin{document}

\maketitle

\begin{abstract}
	\textsf{\fontencoding{OT1}\fontsize{8.5}{10pt}\selectfont{
	We propose an adaptive control strategy for the simultaneous estimation of topology and synchronization in complex dynamical networks with unknown, time-varying topology. Our approach transforms the problem of time-varying topology estimation into a problem of estimating the time-varying weights of a complete graph, utilizing an edge-agreement framework. We introduce two auxiliary networks: one that satisfies the persistent excitation condition to facilitate topology estimation, while the other, a uniform-$\delta$ persistently exciting network, ensures the boundedness of both weight estimation and synchronization errors, assuming bounded time-varying weights and their derivatives. A relevant numerical example shows the efficiency of our methods}.      %
}
\end{abstract}

\section{Introduction}

Dynamical networks, exemplified by a collection of components through a communication network, are increasingly prevalent in various fields, including robotics, autonomous vehicles, distributed computing \cite{mesbahi_graph_2010} and biological systems \cite{moon2015general,luppi2021combining, shilts2022physical}.  The structure of these networks, outlining the interaction patterns among the components, is crucial to shaping the overall behaviour of the networks. However, in many practical scenarios, the topology structure of the network may not be known a priori or be subject to changes, posing a substantial challenge to understanding the fundamental principles for dynamical networks and further control.  
 
There have been many works on addressing network estimation problems, including optimization-based methods, knock-out methods \cite{nabi2012network}, and adaptive control-based methods \cite{zhou2007topology, zhu2021new}, among others, as highlighted in \cite{timme2014revealing}.  Static topology estimation problems are addressed by constructing a synchronized network or by identifying the network by knocking out nodes in \cite{zhou2007topology, nabi2012network, zhu2021new}. As for time-varying topology estimation,  machine learning methods have been applied to estimate network topology, as discussed in \cite{kolar2010estimating}, based on the assumption of either smooth parameter changes or piece-wise constant variations. The unknown switching topology is estimated through adaptive synchronization, specifically under the premise of piece-wise constant changes in switching topology \cite{li2023identifying}. However, these works, including those previously mentioned, primarily focus on the problem of topology estimation, overlooking the application of this topological information in further analysis or control of the network.

When the topology is static and unknown, a combination scheme between topology estimation and control tasks is to identify the topology first, and then use the identified topology for control tasks. After the topology is identified, the network can be controlled for complex tasks by coordination. A combined scheme of topology estimation and control was proposed by switching reference signals in \cite{wang2023finite}. A method in \cite{restrepo2023simultaneous} that realizes topology estimation and synchronization simultaneously was presented by tracking an auxiliary system which synchronizes after identifying topology. However, these methods fail when the topology is time-varying due to their assumptions of static topology.

 This paper proposes an adaptive-control-based method to address the simultaneous topology estimation and synchronization problem for dynamical networks with time-varying topology. The proposed methods guarantee the boundedness of weight estimation and synchronization errors assuming bounded weights and bounded weight derivatives. A scheme of combining the topology estimation and synchronization under time-varying topology is proposed, by estimating the time-varying topology and employing the estimated topology into the control input to synchronize the network.

The structure of the remainder of this paper is as follows: we formulate the 
problem in
Section \ref{sec:problem}.   Section \ref{sec:id} introduces a control scheme and adaptive parameter updating laws for pure topology estimation.  In Section \ref{sec:synch+id1}, we present the solution to the topology estimation and synchronization problem.
Section \ref{sec:simulations} verifies the proposed scheme's effectiveness with a numerical example. Finally, Section \ref{sec:conc} concludes the paper.

\section{Preliminaries}\label{sec:problem}

\subsection{Notations}

 $\mathbb{B}(\Delta)\subset\mathbb{R}^n$ denotes a closed ball of radius $\Delta$ centered at the origin, i.e. $\mathbb{B}(\Delta):=\{x\in\mathbb{R}^n\;:\;|x|\leq \Delta\}$. Denote $\|\cdot\|$  the Euclidean norm of vectors and the induced $L_2$ norm of matrices. The pseudoinverse of a matrix $X$  is denoted as $X^+$.  Denote $|\cdot|$  the absolute value of real numbers.
Denote $\mathcal{G}=(\mathcal{V},\mathcal{E},W)$ a directed weighted graph, where $\mathcal{V}=\{1,2,\dots,N\}$ is a node-set and   $\mathcal{E}\subseteq \mathcal{V}^2$ is an edge set with $M$ edges, characterizing the information exchange between agents.   A directed edge $e_{k}:=(i,j) \in \mathcal{E}$, indicates that agent $j$ has access to information from node $i$, and a positive diagonal matrix $W\in\mathbb{R}^{M\times M}$, whose diagonal $w_k$ entries represent the weights of the edges. We denote time-varying topology as  $\mathcal{G} (t)=(\mathcal{V},\mathcal{E}(t),W(t))$, where the edge set $\mathcal{E}(t)$ and the weight $ w_k(t)$ are time-varying.

\subsection{Model and problem formulation}
We consider a multi-agent system where the agents interact over an \emph{unknown} time-varying topology described by a \emph{directed} graph $\mathcal{G}(t)=(\mathcal V,\mathcal{E}(t), W(t))$, which is assumed to be connected. Without loss of generality, each agent's dynamics is described as follows
\begin{equation}\label{108}
\dot x_i = f_i(x_i) -c\sum_{j=1}^N w_{ij}(t)(x_i-x_j) + u_i \quad i\in\mathcal{V},
\end{equation}
where $x_i\in\mathbb{R}$ is the state of agent $i$; $f_i:\mathbb{R}\to\mathbb{R}$ is a smooth function, denoting its internal dynamics; $w_{ij}(t):\mathbb{R}_{\geq 0} \to\mathbb{R}$ denotes the \emph{unknown} weight function of the interconnection between agents $i$ and $j$; $c$ is a positive constant, denoting the strength of connectivity; $\forall t>0$, $w_{ij}(t)=0$ if the edge $e_{k= (i,j)}\notin\mathcal{E}$ and $w_{ij}(t)\neq0$ if the edge $e_{k= (i,j)}\in\mathcal{E}$. Hence, the edge set $\mathcal{E}(t)$ is time-varying depending on the values of $w_{ij}(t)$. The objective of the multi-agent system \eqref{108} is to achieve consensus among the agents with external control input under an unknown time-varying topology $\mathcal{G}(t)$. The consensus problem considered here can also be extended to formation control or other cooperation tasks.
For each agent's internal dynamics, we assume the following.
\begin{assumption}\label{n182}
	For each agent $i$, there exists a positive constant $L_i$ such that
	\begin{equation}
		\|f_i(x)-f_i(y)\|\leq L_i\|x-y\|
	\end{equation}
	for all $x$, $y\in\mathbb{R}$, where $1\leq i\leq N$.
\end{assumption}

Let $E(t):\mathbb{R}_{\geq 0}\to \mathbb{R}^{N\times M}$ denote the (unknown) incidence matrix function of $\mathcal{G}(t)$ from \cite{mesbahi_graph_2010} and recall that M denotes the number of edges. $E_\odot(t):\mathbb{R}_{\geq 0}\to \mathbb{R}^{N\times M}$ denotes the (unknown) in-incidence matrix function of $\mathcal{G}$, defined as follows: $\left[E_\odot\right]_{ik}(t) := -1$ if $i$ is the terminal node of edge $e_k$ and $\left[E_\odot\right]_{ik}(t) := 0$ otherwise. 
Then, denoting $x:=\left[x_1\;\dots\;x_N\right]^\top$, $F(x):=\left[f_1(x_1)\;\dots\;f_N(x_N)\right]^\top$, and $u:=\left[u_1\;\dots\;u_N\right]^\top$,  \eqref{108} can be written as
\begin{equation}\label{90}
\dot x = F(x)-cE_\odot(t) W(t)E(t)^\top x+u.
\end{equation}

Since the edges of $\mathcal{G}(t)$ are time-varying, the dimension of the incidence matrix function $E(t)$ is not fixed. To represent the unknown time-varying graph, we resort to using a complete graph whose weight of edges is unknown but the number of edges is fixed. 
Denote the incidence matrix $\bar E$ and in-incidence matrix $\bar E_\odot$  of a \emph{complete} graph $\mathcal{K}(\mathcal{V},\mathcal{E}_c,\bar W(t))$, where $\mathcal{E}\subseteq\mathcal{E}_c$. Denote the cardinality of $\mathcal{E}_c$ as  $\bar M$ and $M=N(N-1)$. Let $\bar W(t):=\text{diag}\{\bar w_k(t)\}$ where $\bar w_k(t)\equiv w_k(t)$ if $\bar e_k\in\mathcal{E}$ and $\bar w_k=0$ if $\bar e_k\in\mathcal{E}_c\backslash\mathcal{E}$.  This representation transforms searching for the unknown graph into estimating the weights of edges of the complete graph $\mathcal{K}$. The weight $\bar w_k(t)$ is non-zero if the edge $\bar e_k$ of the complete graph exists in the graph $\mathcal{G}$ to be identified.
Rewrite \eqref{90} as
\begin{equation}\label{99}
\dot x = F(x)-c\bar E_\odot\bar W(t)\bar E^\top x+u.
\end{equation}

\begin{assumption}\label{assumption_weight}
For any $0<k\leq N(N-1)$, there exist upper bounds $w_d$ and $w_d'$ for $\bar {w}_{ij}$ and $\dot {\bar w}_{ij}(t)$ such that
	\begin{equation}
      |\bar w_{k}(t)| \leq w_d,\  \      |\dot {\bar w}_{k}(t)| \leq w_d'
			 \; \; \forall t \ge 0.\end{equation}
\end{assumption}

\begin{remark}
Instead of considering a switching topology, we consider continuous time-varying changes  in the weight of edges here, exploiting the potential robustness of our design in the time-varying topology case.  This assumption contains the cases of adding new edges or removing the old ones by changing the weight of edges in a bounded way. For example, in the human immune cell activation process in response to a pathogen, the concentration level of cytokines, which facilitate communication between immune cells, is smoothly time-varying \cite{talaei2021mathematical}.   This boundedness assumption also means that $w_{ij}$ and its derivative $\dot w_{ij}$ are bounded.
\end{remark}

Using the edge-agreement representation for networked systems with a connected graph enables us to obtain an equivalent reduced system. Defining the edge variable $z:=\bar E^\top x$, rewrite \eqref{99} as
\begin{equation}\label{134}
\dot z = \bar E^\top F(x)-c\bar E^\top \bar E_\odot \bar W(t) z + \bar E^\top u.
\end{equation}
Using suitable labelling of edges, we can partition  the incidence matrix of the complete graph $\mathcal{K}$ as
\begin{equation}
\label{n196} \bar E = \left[\;\bar E_\mathcal{T}\quad \bar E_\mathcal{C}\;\right]
\end{equation}
where  $\bar E_\mathcal{T}\in\mathbb{R}^{N\times(N-1)}$ is the incidence matrix of an spanning tree $\mathcal G_\mathcal{T}\subset \mathcal K$ and $\bar E_\mathcal{C}\in\mathbb{R}^{N\times(\bar M-N+1)}$ denotes the incidence matrix of the remaining edges from \cite{mesbahi_graph_2010}. 
Similarly, partition the edge state as $z=\left[z_\mathcal{T}^\top\;\; z_\mathcal{C}^\top\right]^\top$, where $z_\mathcal{T}\in\mathbb{R}^{(N-1)}$ are the states of the edges of the spanning tree $\mathcal{G}_\mathcal{T}$ and $z_\mathcal{C}\in\mathbb{R}^{\bar M-N+1}$ denote the states of the remaining edges.
Moreover, define 
\begin{equation}\label{n205}
R:=\left[\;I_{N-1}\quad T\;\right],\quad T:=\left(\bar E_\mathcal{T}^\top \bar E_\mathcal{T}\right)^{-1}\bar E_\mathcal{T}^\top \bar E_\mathcal{C},
\end{equation}
with $I_{N-1}$ denoting the $N-1$ identity matrix. Based on \eqref{n196} and \eqref{n205}, we have
$\bar E = \bar E_\mathcal{T} R$ and
$z=R^\top z_\mathcal{T}$.
Then, we obtain  a reduced-order model  of \eqref{134} as 
\begin{equation}\label{n231}
\dot z_\mathcal{T} = \bar E_\mathcal{T}^\top F(x)-c\bar E_\mathcal{T}^\top \bar E_\odot \bar W(t) R^\top z_\mathcal{T} + \bar E_\mathcal{T}^\top u.
\end{equation}

 The topology estimation problem in \eqref{108} is transformed into estimating the time-varying diagonal entries of the matrix function $\bar{W}(t)$ in \eqref{n231}. Meanwhile, the synchronization problem for \eqref{108} is transformed into the stabilization problem of the origin for the reduced-order system \eqref{n231}.

\section{Topology estimation under bounded time-varying weights}\label{sec:id}

In this section, we introduce the external input $ u(t)$ to estimate the unknown graph topology $\mathcal{G}(t)$ for the dynamical systems \eqref{n231}.  A refined control design to our previous work addressing static topology estimation \cite{restrepo2023simultaneous} will be used.

\subsection{Control design and weight estimation laws}
Denote $\bar w(t):=\left[\bar w_1(t) \;\cdots\; \bar w_{\bar M}(t)\right]^\top\in\mathbb{R}^{\bar M}$ as the vector of unknown weights, $\hat w(t):=\left[\hat w_1(t) \;\cdots\; \hat w_{\bar M}(t)\right]^\top\in\mathbb{R}^{\bar M}$ as its estimate, and $\hat W(t):=\text{diag}\{\hat w(t)\}$.

Set the updating law 
\begin{equation}\label{170}
\dot{\hat w} =  -c\hat Z(t)\bar E_\odot^\top \bar E_{\mathcal{T}}  \tilde{z}_{\mathcal{T}},
\end{equation}
where $\tilde{z}_\mathcal{T}:=z_\mathcal{T}-\hat{z}_\mathcal{T}=\bar E_\mathcal{T}^\top \tilde x(t)$, $\tilde x(t):=x(t)-\hat{x}(t)$, $\hat{z}(t):=\bar E^\top \hat x(t)$, $\hat Z(t):=\text{diag}\{\hat z(t)\}$, and $\hat x(t)$ is an auxiliary variable to be designed later.

Select the control input
\begin{equation}\label{166}
u = -c_1 (x-\hat x(t)) + \dot{\hat{x}}(t) + c\bar E_\odot \hat W(t) \hat z(t) - F(\hat x(t))
\end{equation}
where $c_1$ is a positive constant.

\subsection{Time-varying topology estimation}
In this part, we analyze the effect of the time-varying weights on the topology estimation and show that using our design, the weight estimation errors remain bounded.

 Define $\tilde{w}(t):=\bar w (t)- \hat w(t)$. Utilizing  \eqref{n231}, \eqref{170} and \eqref{166}, we derive the closed-loop system as 
\begin{equation}\label{184}\begin{aligned}[b]
\left[\begin{matrix}
\dot{\tilde z}_\mathcal{T}\\\dot{\tilde w}
\end{matrix}\right]\!\!=\!\!&\left[\begin{matrix}
-c_1I-c\bar L_e & -c\bar{E}_\mathcal{T}^\top\bar{E}_\odot\hat{Z}(t)\\c\hat{Z}(t)\bar{E}_\odot^\top\bar{E}_\mathcal{T}&0\!\! 
\end{matrix}\right]\!\!\left[\begin{matrix} \tilde z_\mathcal{T}\\\tilde w\end{matrix}\right]\\
&+\left[\begin{matrix}\bar E_\mathcal{T}^\top\tilde F(x,\hat x)\\ \dot{\bar w}\end{matrix}\right],
\end{aligned}\end{equation}
where  $\bar L_e := \bar{E}_\mathcal{T}^\top\bar{E}_\odot\bar W R^\top$ and $\tilde F(x,\hat x):=F(x)-F(\hat x(t))$.

\begin{proposition}\label{prop1}
	Assume that the signal $\hat{Z}(t)$ is bounded, globally Lipschitz and satisfies that for any unit vector $v \in \mathbb{R}^{\bar M}$
 \begin{align}
 \label{n2011}
     \int_{t}^{t+T} \!\!\! & \| \hat Z(\tau) v\|^2 d\tau >\mu, \forall t\geq0. 
 \end{align}
 where $T, \mu>0$.
 With Assumptions~\ref{n182} and \ref{assumption_weight},  the edge weight estimation errors  $\tilde w(t)$ of the multi-agent system \eqref{108}  are globally ultimately bounded, and all the closed-loop signals are bounded, after applying update law \eqref{170} and the control input \eqref{166}. 
\end{proposition}
\begin{proof}
	The closed-loop system \eqref{184} can be regarded as a perturbed form of 
	\begin{equation}\label{n300}
	\left[\begin{matrix}
	 \dot{\tilde z}_\mathcal{T}\\\dot{\tilde w}
	\end{matrix}\right]\!\!=\!\!\left[\begin{matrix}
	-c_1I-c\bar L_e &\!\! -c\bar{E}_\mathcal{T}^\top\bar{E}_\odot\hat{Z}(t)\\c\hat{Z}(t)\bar{E}_\odot^\top\bar{E}_\mathcal{T}&\!\! 0
\end{matrix}\right]\!\!\!\left[\begin{matrix} \tilde z_\mathcal{T}\\\tilde w\end{matrix}\right]\!\!.
	\end{equation}
 Since the graph is assumed to be connected, the eigenvalues of edge Laplacian $\bar L_e$  have positive real parts from \cite{mukherjee_robustness_2018}. Hence, $-c_1I-c\bar L_e$ is Hurwitz. And $\bar E_\mathcal{T}^\top \bar E_\odot$ has rank $N-1$ since $\bar E_\mathcal{T}^\top \bar E_\odot R^\top $ is full rank, as discussed in \cite{mukherjee_robustness_2018}. Then $(-c_1I-c\bar L_e \!\!\bar{E}_\mathcal{T}^\top\bar{E}_\odot)$ is controllable. 
 
  If $(-c_1I-c\bar L_e \bar{E}_\mathcal{T}^\top\bar{E}_\odot)$ is controllable, and $Z(\tau) $ is piecewise-continuous, bounded and satisfies \eqref{n2011},  then global uniform exponential stability of the origin for \eqref{n300} follows from Theorem 5 \cite{morgan1977stability} or Theorem 2.17 \cite{narendra2012stable, NARENDRA1987PE} and the linearity of \eqref{n300}.
	
	Denote $\xi:=\left[\tilde{z}_\mathcal{T}^\top\;\;\tilde w^\top\right]^\top\in\mathbb{R}^{\bar M + (N-1)}$. From the global exponential stability of \eqref{n300} and from converse Lyapunov theorems (Theorem 4.14 of \cite{Khalil2002nonlinear}), there exists a Lyapunov function $V(t,\xi):\mathbb{R}_{\geq 0} \times \mathbb{R}^{\bar M + (N-1)}\to\mathbb{R}_{\geq 0}$ such that
	\begin{IEEEeqnarray}{rCl}
 \label{210}
		\beta_1\|\xi\|^2\leq V(t,\xi)&\leq& \beta_2\|\xi\|^2\\
		\left\|\frac{\partial V}{\partial \xi}\right\|\leq \beta_3\|\xi\|,
	\end{IEEEeqnarray} 
	for some $\beta_1,\beta_2,\beta_3>0$, and its derivative along the trajectories of \eqref{n300} satisfies
	\begin{equation}
	\dot V(t,\xi) \leq -\beta_4\|\xi\|^2, \quad \beta_4>0.
	\end{equation} 

 In view of Assumption~\ref{n182}, we can further obtain
\begin{equation}\label{n285}
	\left\Vert\left[\bar E_\mathcal{T}^\top\left[F(x)-F(\hat x(t))\right]\right]_k\right\Vert\leq L_f\|\tilde z_k\|,
\end{equation}
 where $L_f:=\max\limits_{i\in\mathcal{V}}\{L_i\}$. Choose
\begin{equation}
\label{270}
    V_1(t,\xi)=0.5\|\tilde z_{\mathcal{T}}\|^2+0.5\|\tilde w\|^2.
\end{equation}
Along the trajectories of \eqref{184}, its derivative is 
\begin{equation}\begin{aligned}
\label{100}
    \dot V_1(t,\xi)=&-\tilde z_{\mathcal{T}}^\top(c_1I+c\bar L_e))\tilde z_{\mathcal{T}} -\tilde z_{\mathcal{T}}^\top\bar E_\mathcal{T}^\top\tilde F(x,\hat x) +\tilde w^\top \dot {\bar w}\\
    \leq &-(c_1+c\lambda_{\min}{\bar L_e}-L_f)\|\tilde z_{\mathcal{T}}\|^2+\|\tilde w\|\|\dot {\bar w}\|\\
    = &-c_1'\|\tilde z_{\mathcal{T}}\|^2+\|\tilde w\|\|\dot {\bar w}\|,
    \end{aligned}
\end{equation}
where $c_1':=c_1+c\lambda_{\min}\{\bar L_e\}-L_f$ and $\lambda_{\min}\{\bar L_e\}$ is the smallest eigenvalue of ${\bar L_e}$. The second inequality is obtained using \eqref{n285}.

Let $V'(t,\xi)=V(t,\xi)+V_1(t,\xi)$. In view of \eqref{n285}, \eqref{100} and Assumption 2, its derivative along the trajectories of \eqref{184} is 
\begin{equation}\label{n331}\begin{aligned}[b]
	\dot V'(t,\xi)\leq& -\beta_4\|\xi\|^2 + \frac{\partial V}{\partial \tilde z_{\mathcal{T} } }^\top \bar E_\mathcal{T}^\top \tilde{F}(x,\hat x)+\frac{\partial V}{\partial \tilde w}^\top  \dot {\bar w}\\
 &-c_1'\|\tilde z_{\mathcal{T}}\|^2+\|\tilde w\|\|\dot {\bar w}\|\\
 \leq&  \frac{\delta^2}{4}\left(\left\| \frac{\partial V}{\partial \tilde z_{\mathcal{T} } }\right\|^2+\left\|\frac{\partial V}{\partial \tilde w}\right\|^2\right)+\frac{L_f^2\|\tilde z_{\mathcal{T} }\|^2}{\delta^2}\\
 &-\beta_4\|\xi\|^2  +2\|\dot {\bar w}\|^2/{\delta^2}+\delta^2\|\tilde w\|^2/{4}-c_1'\|\tilde z_{\mathcal{T}}\|^2\\
	\leq& -\beta_4\|\xi\|^2 + \beta_3^2\delta^2\|\xi\|^2/{4}+\delta^2\|\xi\|^2/{4}+2{w'_d}^2/{\delta^2}\\
	\leq& -\beta_4'\|\xi\|^2+\beta_5,
	\end{aligned} \end{equation} 
 where $\beta_4':=\beta_4-\beta_3^2 \delta^2/4-\delta^2/4$, $\beta_5:=2{w_d'^2}/{\delta^2}$ and $\delta >0$, and we choose $c_1'$ that satisfies $c_1'-L_f^2/{\delta^2}>0$. The second inequality is obtained by applying Young's inequality. Then, by properly choosing $V$ and $\delta$ such that $\beta_4'>0$,  the origin of \eqref{184} is globally ultimately bounded from  \eqref{210}, \eqref{270} and \eqref{n331} by Theorem 4.18 in \cite{Khalil2002nonlinear}. The estimation error $\|\tilde w\|$ is globally ultimately bounded  and converges to $\Omega_{\tilde w}:=\{\tilde w:\|\tilde w\|\leq d_{\tilde w}\} $ with $d_{\tilde w}=\sqrt{{\beta_5}/{\beta'_4}} $. By \eqref{170}, \eqref{166}, $u(t)$ and $\hat w(t)$ are also bounded. Hence, the result follows.
\end{proof}

\begin{remark}
Proposition~\ref{prop1} shows that for dynamical systems \eqref{108} with time-varying topology, the control input \eqref{166} and weight estimation law \eqref{170} guarantee the boundedness of the weight estimation errors  $\|\tilde w\|$ provided that $\hat Z(t)$ is persistently exciting. Besides, if the weights are fixed, then the estimation errors $\|\tilde w\|$ will be bounded and further converge to zero.
\end{remark}

\section{Simultaneous topology estimation and synchronization for time-varying networks}\label{sec:synch+id1}

In this section, we explore simultaneous topology estimation and synchronization for \eqref{99} with the time-varying topology. We use the control input \eqref{166} in the following scheme. 

\subsection{Design of updating laws and auxiliary system}
Let $\hat{z}(t)$ be the state of an auxiliary dynamical system. Set the new updating law instead of \eqref{170} as
\begin{equation}\label{170_n}
\dot{\hat w} =  -c\hat Z(t)\bar E_\odot^\top \bar E_{\mathcal{T}} \tilde{z}_{\mathcal{T}}-\sigma_1 \hat w,
\end{equation}
where $\sigma_1$ is a positive constant. 

The updating law \eqref{170_n} adds  $\sigma_1 \hat w$ to guarantee the boundedness of $\|\tilde { w}\|$  under the bounded derivation of $\|w\|$.

Design the auxiliary dynamical system  as 
\begin{equation}\label{313}
	\dot{\hat z} = \bar E^\top F(\hat x) -c_2\hat{z} + \psi(t,\tilde z_\mathcal{T})
	\end{equation}
where $c_2>L_f$ is a positive constant and  the function $\psi(t,\tilde z_\mathcal{T}): \mathbb{R}_{\geq 0}\times\mathbb{R}^{N-1}\to\mathbb{R}^{\bar M}$ satisfies that
	\begin{equation}
	\max\left\{\|\psi(\cdot)\|,\left\|\frac{\partial \psi(\cdot)}{\partial t}\right\|,\left\|\frac{\partial \psi(\cdot)}{\partial \tilde z}\right\|\right\}\leq \kappa(\|\tilde z_\mathcal{T}\|), \forall t\geq 0,\label{258}
	\end{equation}
where $\kappa:\mathbb{R}_{\geq 0}\to\mathbb{R}_{\geq 0}$ is a continuous non-decreasing function. Define  $\Psi(t,\tilde z_\mathcal{T}) \in R^{\bar M\times \bar M}$ as a diagonal matrix function of $\psi(\cdot)$. Specifically, write $\Psi(t,x_1):=\text{diag}\{\psi(t,x_1)\}$. 
Define $\Psi'(t,\tilde z_\mathcal{T}):\mathbb{R}_{\geq 0}\times\mathbb{R}^{N-1}\to\mathbb{R}^{(N-1) \times \bar M}$. Let $\Psi'(t,\tilde z_\mathcal{T})=\bar{E}_\mathcal{T}^\top\bar{E}_\odot \Psi(t,\tilde z_\mathcal{T})$. $ \Psi'(t,\tilde z_\mathcal{T})$
 is uniform $\delta$-persistently exciting (u$\delta$-PE) with respect to $\tilde z_\mathcal{T}$ as per Definition 5 in \cite{panteley2001relaxed}.

\subsection{Stability analysis of the unperturbed systems}
Using \eqref{n231}, \eqref{166} and \eqref{170_n}, we obtain the new closed-loop system as
\begin{equation}\label{closesystem}
 \begin{aligned}
 \left[\begin{matrix}
 \dot{\tilde z}_\mathcal{T}\\\dot{\tilde w}
\end{matrix}\right]=&\left[\begin{matrix} \bar E_\mathcal{T}^\top\tilde F(x,\hat x) -(c_1I+c\bar L_e )\tilde z_\mathcal{T}-c\bar{E}_\mathcal{T}^\top\bar{E}_\odot\hat{Z}(t)  \tilde w \\  c\hat{Z}(t)\bar{E}_\odot^\top\bar{E}_\mathcal{T} \tilde z_\mathcal{T} \end{matrix}\right]\\
&+\left[\begin{matrix}0\\ \dot{\bar w} +\sigma_1 \hat w \end{matrix}\right].
\end{aligned}\end{equation}

Similar to the previous analysis,  the closed-loop system \eqref{closesystem} can be seen as the perturbed version of 
\begin{equation}\label{closesystem_un}
 \begin{aligned}
 \left[\begin{matrix}
 \dot{\tilde z}_\mathcal{T}\\\dot{\tilde w}
\end{matrix}\right]=&\left[\begin{matrix} \bar E_\mathcal{T}^\top\tilde F(x,\hat x)-(c_1I+c\bar L_e )\tilde z_\mathcal{T}-c\bar{E}_\mathcal{T}^\top\bar{E}_\odot\hat{Z}(t)  \tilde w \\ c\hat{Z}(t)\bar{E}_\odot^\top\bar{E}_\mathcal{T}RR^\top \tilde z_\mathcal{T} \end{matrix}\right].
\end{aligned}\end{equation}

Before studying the stability of  \eqref{closesystem}, we first analyze the stability of the unperturbed system \eqref{closesystem_un}.  Replacing $x_1$, $x_2$, $A(t,x_1)$, $B$, $\Phi$ and $\phi$ by $\tilde z_\mathcal{T}$, $\tilde w$, $ \bar E_\mathcal{T}^\top\tilde{F}(x,\hat x)-\left(c_1I+\bar E_\mathcal{T}^\top \bar E_\odot \bar W R^\top\right)\tilde{z}_\mathcal{T} $, $ \bar E_\mathcal{T}^\top \bar E_\odot$, $\hat Z$ and $\hat z$, respectively, 
 we can represent \eqref{closesystem_un} as 
	\begin{equation}\label{273}
	\left[\begin{matrix}
	\dot{x}_1\\\dot{x}_2
\end{matrix}\right]=\left[\begin{matrix}
	A(t,x_1) + B\Phi(t,x_1)^\top x_2\\-\Phi(t,x_1)B^\top x_1
	\end{matrix}\right]
	\end{equation}
	where $x^\top:=\left[x_1^\top\;\;x_2^\top\right]$,  $\Phi(t,x_1):\mathbb{R}_{\geq 0}\times\mathbb{R}^{N-1}\to\mathbb{R}^{\bar M\times \bar M}$  and  $\phi(t,x_1):\mathbb{R}_{\geq 0}\times\mathbb{R}^{N-1}\to\mathbb{R}^{\bar M}$ are piece-wise continuous in $t$ and continuous in $x_1$. Moreover, $\Phi(t,x_1)$ is diagonal with    $\Phi(t,x_1):=\text{diag}\{\phi(t,x_1)\}$.    
 Assume the following:
	\begin{assumption}\label{n422}
		The function $A$ is locally Lipschitz in $x$ uniformly in $t$. Moreover, there exists a continuous nondecreasing function $\rho_1:\mathbb{R}_{\geq 0}\to\mathbb{R}_{\geq 0}$ such that $\rho_1(0)=0$ and for all $(t,x_1)\in\mathbb{R}\times\mathbb{R}^{ N-1}$, $\|A(t,x_1)\|\leq \rho_1(\|x_1\|)$.
	\end{assumption}
	\begin{assumption}\label{n425}
		There exists a locally Lipschitz function $V_1:\mathbb{R}_{\geq 0}\times\mathbb{R}^{N+M-1}\to\mathbb{R}_{\geq 0}$, and $\alpha_1, \alpha_2, \alpha_3 >0$ such that
		\begin{equation}
			\alpha_1\|x\|^2\leq V_1(t,x)\leq \alpha_2\|x\|^2\label{n428}
		\end{equation}
		and its derivative along the trajectories of \eqref{273} satisfies
		\begin{equation}
			\dot V_1(t,x) \leq -\alpha_3\|x_1\|^2.\label{n429}
		\end{equation} 
		
	\end{assumption}
	Then we state the following lemma.
	\begin{lemma}\label{USPAS}
 Let Assumptions~\ref{n422} and \ref{n425} hold. Assume $B\Phi^\top (t,x_1)$ is u$\delta$-PE with respect to $x_1$ and   $\Phi(t,x_1)$  satisfies 
 \begin{equation}
	\max\left\{\|\Phi(\cdot)\|,\left\|\frac{\partial \Phi(\cdot)}{\partial t}\right\|,\left\|\frac{\partial \Phi(\cdot)}{\partial x_1}\right\|\right\}\leq\rho(\|x_1\|), \forall t\geq 0,\label{259}
	\end{equation}
where $\rho:\mathbb{R}_{\geq 0}\to\mathbb{R}_{\geq 0}$ is a continuous non-decreasing function.
 Then the origin of \eqref{273} is \emph{uniformly semiglobally asymptotically stable}. 
	\end{lemma}
	
 	\begin{proof}
Consider a Lyapunov function candidate as 
\begin{equation}\label{303}
    \begin{aligned}
        V(t,x)&:=V_1(t,x)+\varepsilon V_4(t,x)\\
        V_4(t,x)&:=V_2(t,x)+V_3(t,x)\\
	V_2(t,x)&:=-x_1^\top B\Phi(t,x_1)^\top x_2\\
	V_3(t,x)&:=-\int_{t}^{\infty}e^{(t-\tau)}\left\|B\Phi(\tau,x_1)^\top x_2\right\|^2d\tau,
	     \end{aligned}\end{equation}
where $V_1(t,x)$ is given in Assumption~\ref{n425} and $\varepsilon>0$.  
 Using the u$\delta$-PE of $B\Phi^\top$, for all $(t,x)\in\mathbb{R}\times\mathbb{B}(\Delta)$, one has 
 \begin{equation} \label{v3b}
 \begin{aligned}
     V_3(t,x)=& -\int_{t}^{\infty}e^{(t-\tau)}  x_2^\top\Phi(\tau,x_1)B^\top B\Phi(\tau,x_1)^\top x_2  d\tau\\
     \leq& -\int_{t}^{t+T}e^{(t-\tau)} x_2^\top\Phi(\tau,x_1)B^\top B\Phi(\tau,x_1)^\top x_2 d\tau\\
   \leq  &(e^{-T}-1) \mu \|x_2\|^2,
 \end{aligned}    
 \end{equation} 
where $b':=(e^{-T}-1) \mu$, $\mu$ and $T$ are defined from Definition 5 in \cite{panteley2001relaxed}. 
In view of \eqref{259} and \eqref{v3b}, $ V_4(t,x)$ in \eqref{303} satisfies, for all $(t,x)\in\mathbb{R}\times\mathbb{B}(\Delta)$, 
	\begin{equation}\label{326}
	V_4(t,x)\leq b\|x_1\|\rho(\|x_1\|)\|x_2\|- b'\|x_2\|^2,
	\end{equation} 
	where $b:=\|B\|$.  Define $b_\rho:=b\rho(\Delta)$. In view of \eqref{326}, $\varepsilon V_4(t,x)$ satisfies on $\mathbb{R}\times\mathbb{B}(\Delta)$  
	\begin{equation}\label{333}\begin{aligned}
	-\varepsilon\rho(\Delta)\|x_2\|^2&-\varepsilon b_\rho\|x_1\|\|x_2\|\leq \varepsilon V_4(t,x)\leq \varepsilon b_\rho\|x_1\|\|x_2\|\\&-\varepsilon  b'\|x_2\|^2.
	\end{aligned}\end{equation}
	So, from \eqref{n428} and \eqref{333}, for any $\Delta>0$ and for a sufficiently small $\varepsilon$, there exist $\underline{\alpha}_{\Delta}>0$ and $\overline{\alpha}_{\Delta}>0$ such that for all $(t,x)\in\mathbb{R}\times\mathbb{B}(\Delta)$
	\begin{equation}\label{341}
	\underline{\alpha}_{\Delta}\|x\|^2\leq V(t,x) \leq \overline{\alpha}_{\Delta}\|x\|^2.
	\end{equation}
 
	We proceed to obtain the derivative of $V_4(t,x)$ along the trajectories of the system \eqref{273}. First, we have
	 \begin{equation}\label{v2}\begin{aligned}
	\dot V_2(t,x)=& \|\Phi(t,x_1)B^\top x_1\|^2-x_2^\top \Phi(t,x_1)^\top B^\top A(t,x_1)\\
	&\; -\| B \Phi(t,x_1) x_2\|^2-x_2^\top \dot{\overbrace{\Phi(t,x_1)}}B^\top x_1,
	\end{aligned}\end{equation} 
 where  $\dot{\overbrace{\Phi(t,x_1)}}:=\frac{\partial \Phi(t,x_1)}{\partial t}+\frac{\partial \Phi(t,x_1)}{\partial x_1}$.  
	Next, we have
	 \begin{IEEEeqnarray*}{rCl} \label{v3}
	\frac{\partial V_3}{\partial x_1}&=&-\int_{t}^{\infty}\!\!\!\!{2e^{(t-\tau)}x_2^\top\Phi(\tau,x_1)B^\top B\left[\frac{\partial \Phi(\tau,x_1)}{x_1}^\top x_2\right]} d\tau\\
	\frac{\partial V_3}{\partial x_2}&=&-\int_{t}^{\infty}{2e^{(t-\tau)}\Phi(\tau,x_1)B^\top B\Phi(\tau,x_1)^\top x_2} d\tau\\
	\frac{\partial V_3}{\partial t}&=&\!\left\|B\Phi(t,x_1)^\top \!x_2\right\|^2 \!\!-\!\! \int_{t}^{\infty}\!\!\!\frac{\partial}{\partial t}\left[e^{(t-\tau)}\!\!\left\|B\Phi(\tau,x_1)^\top \!x_2\right\|^2\right]\! d\tau.
	\end{IEEEeqnarray*} 

	From Assumption~\ref{n422}, \eqref{259}, \eqref{v2} and \eqref{v3}, we obtain an upper bound function for the derivative of $V_4(t,x)$. Define $$\bar\rho(r,s):=b_\rho\left[(2+2b^2_\rho)rs+(1+2b^2_\rho)\rho_1(r)s+b_\rho r^2+2b^2_\rho s^2\right].$$ Then, for $(t,x)\in\mathbb{R}\times\mathbb{B}(\Delta)$, 
	 \begin{equation}\label{318}
	\dot V_4(t,x)\leq \bar\rho(\|x_1\|,\|x_2\|) - b'\|x_2\|^2.
	\end{equation}

	Using \eqref{n429} and \eqref{318}, the derivative of $V(t,x)$ satisfies, for all $(t,x)\in\mathbb{R}\times\mathbb{B}(\Delta)$, 
	 \begin{equation*}\begin{aligned}[b]
	\dot V(t,x) \leq& -\alpha_3\|x_1\|^2 - \varepsilon (2b_\rho+2b^3_\rho)\|x_1\|\|x_2\| + 2\varepsilon b_\rho^3\|x_2\|^2\\&  + \varepsilon b_\rho^2\|x_1\|^2 +(b_\rho+2b^3_\rho)\rho_1(\|x_1\|)\|x_2\|- \varepsilon b'\|x_2\|^2.
	\end{aligned}\end{equation*} 
Note that $b'=(e^{-T}-1) \mu$.	Choosing $\mu$ and $T$ such that  $b' \geq b_\rho^2+2b_\rho^3+\beta'$ and $\beta'>0$ yields
	\begin{equation*}\begin{aligned}[b]
	\dot V(t,x) 	\leq& -\left(\alpha_3-\left(4+b_\rho^2+4b_\rho^4\right)\varepsilon \right)\|x_1\|^2- \varepsilon \beta'\|x_2\|^2 \\
  &+(1+4b_\rho^4)\varepsilon\rho^2_1(\|x_1\|). 
	\end{aligned}\end{equation*}
 	Selecting $\varepsilon$ sufficiently small such that $\alpha_3-\varepsilon\left(4+b_\rho^2-4b_\rho^4\right) -\varepsilon(1+4b_\rho^4) {\rho_1(|\Delta|)/|\Delta|}^2 >\alpha$,
 yields
  \begin{equation}\begin{aligned}[b] \label{353}
	\dot V(t,x)   \leq& -\alpha\|x_1\|^2- \beta\|x_2\|^2,
	\end{aligned}\end{equation}
 where $\beta=\varepsilon \beta'$.
	Therefore, by Theorem 4.9 in \cite{Khalil2002nonlinear}, for all $(t,x)\in\mathbb{R}\times\mathbb{B}(\Delta)$,  the origin of \eqref{273} is semi-globally uniformly asymptotically stable from \eqref{341} and \eqref{353}. 
	\end{proof}

\begin{remark}
Contrary to \cite{restrepo2023simultaneous} which studies the stability where the unknown parameters are defined in a certain set,  Lemma ~\ref{USPAS} analyzes the stability for \eqref{273} when the parameters are unknown and fixed. The result from  Lemma ~\ref{USPAS} yielding uniform global asymptotical stability is thus stronger than the case of uniform practical stability derived in \cite{restrepo2023simultaneous}.
\end{remark}

\subsection{Simultaneous topology estimation  and synchronization}
Considering the time-varying weights as the disturbance of system \eqref{273}, we analyze the robustness of system \eqref{closesystem} under time-varying topology in Proposition 2,  based on Lemma~\ref{USPAS}. 

\begin{proposition}\label{propUSPAS}
	 Let Assumptions~\ref{n182} and \ref{assumption_weight} hold.
Then, the origin of the closed-loop system \eqref{closesystem} with the update law \eqref{170_n} and control input \eqref{166},
	 is \emph{uniformly semi-globally stable} with $\hat z(t )$ given by the update law \eqref{313}. Its weight estimation errors $\tilde w$ are ultimately bounded, and converge to a set $\Omega_{\tilde w}$.  Furthermore, the edge states $z$ are also ultimately bounded, and converge to a set $\Omega_z$.
\end{proposition}

 The sets $\Omega_{\tilde w}$ and $\Omega_z$ are defined in the proof that follows.

\begin{proof}
We first show that $\hat z(t)$ is u$\delta$-PE with respect to $\tilde z_\mathcal{T}$. Denote $\xi:=\left[\tilde{z}_\mathcal{T}^\top\;\;\tilde w^\top\right]^\top$. Define $V_1(t,\xi)$ as in \eqref{270}. Its derivative along \eqref{closesystem} is
\begin{equation}\label{231}\begin{aligned}
	\dot V_1(t,\xi)=&-\tilde z_{\mathcal{T}}^\top(c_1I+c\bar L_e))\tilde z_{\mathcal{T}} -\tilde z_{\mathcal{T}}^\top\bar E_\mathcal{T}^\top\tilde F(x,\hat x) +\tilde w^\top \dot {\bar w}\\&+ \sigma_1 \dot {\bar w}^\top \hat w\\
 \leq&-(c_1+c\lambda_{\min}{\bar L_e}))\|\tilde z_{\mathcal{T}}\|^2-\tilde z_{\mathcal{T}}^\top\bar E_\mathcal{T}^\top\tilde F(x,\hat x) \\
 &-\sigma_1 \|\tilde w\|^2+\|\tilde w\|\|\dot w\|+\sigma_1\|\tilde w\|\| \bar w\|\\
  \leq&-c_1'\|\tilde z_\mathcal{T}\|^2-\sigma_1' \|\tilde w\|^2+d\\
  \leq& -c_1''\|\xi\|^2+d,
\end{aligned}
	\end{equation}
  where $ \quad c_1'$ is defined in \eqref{100}, $\sigma_1':=\sigma_1-0.5(\sigma_1+1)/{\delta^2}>0$, $d:={0.5 \delta^2(\sigma_1 | w_d|^2}+|w_d'|^2)$ and $c_1''=\min\{c_1',\sigma_1'\}$. 
   
	From \eqref{270} and \eqref{231}, the system \eqref{closesystem} is globally uniformly stable \cite{Khalil2002nonlinear} and $\xi$ converges to the set $\Omega:=\{ \xi: \|\xi\| \leq \sqrt{{d}/{c_1''} } \}$. Therefore, the solutions $\xi(t)$ are ultimately bounded from Theorem 4.18 in \cite{Khalil2002nonlinear}.
	
	Choose the Lyapunov function $ V_5(\hat{z}):=0.5\|\hat{z}\|^2$. Its derivative \eqref{313} along the trajectories of the auxiliary system \eqref{313} satisfies	  \begin{equation}\label{237}\begin{aligned}
	\dot{ V}_5(\hat z)=&-c_2\hat z^\top \hat z+\hat z^\top \bar E^\top F(\hat x)+\hat z^\top \psi(t,\tilde z_{\mathcal{T}})\\
 \leq&-c_2\|\hat z\|^2+L_f\|\hat z\|^2+\|\hat z\|\|\psi(t,\tilde z_{\mathcal{T}})\|\\
	\leq&-c_2'|\hat z|^2+|\kappa(\|\tilde z_{\mathcal{T}}\|)|^2	\leq-c_2'\|\hat z\|^2+\sigma,
	\end{aligned}\end{equation}
where $c_2':=c_2-L_f-0.25$.   The third inequality is obtained by \eqref{258} and Young's inequality. As the solution $\tilde z_{\mathcal{T}}(t)$ is uniformly  stable, there exists  $\sigma>0$ such that  $|\kappa(\|\tilde z_{\mathcal{T}}\|)|^2\leq\sigma$  for all $t\geq0$. The last inequality follows.
	Similarly, from \eqref{237} the solutions $\hat z(t)$, are ultimately bounded. 
 
 In Lemma~\ref{lemma1} in Appendix~\ref{apx:dPE},   $x$ and $w$ in \eqref{app_sys} correspond to $[{\tilde z}_\mathcal{T}^\top\;\;   \tilde w^\top]^\top$ and $\hat z$ respectively here. From \eqref{313}, \eqref{258} and \eqref{closesystem}, the inequalities \eqref{n751} and \eqref{n752} hold. \eqref{313} implies $f_1(t,w)\leq l\|w\|$ with $l:=c_2+L_f$ in Lemma~\ref{lemma1}.  Based on the boundedness of $\xi$ and $\hat z(t)$, \eqref{n755} holds. Now, since all the assumptions in Lemma~\ref{lemma1} in Appendix~\ref{apx:dPE} are satisfied,     $B\hat Z(t)$, given by the update law \eqref{313}, is u$\delta$-PE with respect to $\tilde z_\mathcal{T}$. 

	Next, we will analyze the stability of \eqref{closesystem}.   
      For $A(t,x_1)$ in \eqref{closesystem}, there exists a function  $\rho_1(\|x_1\|):=k\|x_1\|$ where $k:=\max\{ L_f+\|c_1I+ \bar E_\odot \bar W \bar E\|, L_f+\| c_1I+\bar E_\mathcal{T}^\top \bar E_\odot \bar W R^\top\|\}$, such that Assumption~\ref{n422} is satisfied. 
      Along the trajectories of \eqref{closesystem_un}, the derivative of $V_1(t,\xi)$ defined as \eqref{270} is 
 \begin{equation}
    \label{lya_v1}
    \begin{aligned}       
 \dot V_1(t,\xi)=&-(c_1I+c{\bar L_e}))\tilde z_{\mathcal{T}}^\top \tilde z_{\mathcal{T}}-\tilde z_{\mathcal{T}}^\top\bar E_\mathcal{T}^\top\tilde F(x,\hat x)\\
    \leq &-c_1'\|\tilde z_{\mathcal{T}}\|^2.
     \end{aligned}
\end{equation} where $c_1'$ is defined in \eqref{100}. Hence, $V_1(t,\xi)$ satisfies Assumption ~\ref{n425} with $\alpha_1=\alpha_2:=\frac{1}{2}$ and $\alpha_3:=c_1'$. Now that all the assumptions of Lemma~\ref{USPAS} hold,  the origin of system \eqref{closesystem_un} is concluded to be uniformly asymptotically stable.

     Consider again $V(t,\xi)$ defined in \eqref{303}. In order to notationally distinguish the derivatives of $V_i(t,\xi)$  along the trajectories of  \eqref{closesystem} and \eqref{closesystem_un}, we denote $\dot V'_i(t,\xi)$ as the derivative of $V_i(t,\xi)$ for \eqref{closesystem}  while $\dot V_i(t,\xi)$ corresponds to \eqref{closesystem_un}, where $i=1,2,3,4$. Denote  $\Delta V'_i=\dot V'_i(t,\xi)-\dot V_i(t,\xi)$.
     Based on \eqref{closesystem}, \eqref{closesystem_un},  \eqref{v2}, \eqref{v3} and \eqref{lya_v1},  $\Delta V'_i(t,\xi)$ is  
     \begin{equation}   \label{n111}  \begin{aligned}
        \Delta V'_1=&-\sigma_1 \| \tilde w\|^2+ \tilde w^\top\dot {\bar w}+\sigma_1  \tilde w^\top\bar w\\
            \Delta V_2'=& -\tilde z_\mathcal{T}^\top B \Phi(t,\tilde z_\mathcal{T})^\top( \dot {\bar w}+\sigma_1 \bar w-\sigma_1 \tilde w)\\
             \Delta V_3'= &-\int_{t}^{\infty}{2e^{(t-\tau)}\Phi(\tau,\tilde z_\mathcal{T})B^\top B\Phi(\tau,\tilde z_\mathcal{T})^\top \tilde w} d\tau\\
            & \cdot ( \dot {\bar w}+\sigma_1 \bar w-\sigma_1 \tilde w).
        \end{aligned} \end{equation} 
According to  \eqref{closesystem_un} and \eqref{n111}, using Young's inequality, we have
  \begin{equation} \label{n211}\begin{aligned}
        \Delta V'_1\leq&-(\sigma_1-\delta^2-\delta^2 \sigma_1^2)| \tilde w|^2+\frac{\|\dot {\bar w}\|^2}{4\delta^2}+\frac{\| {\bar w}\|^2}{4\delta^2}.\\
        \varepsilon \Delta V'_2\leq&\varepsilon(b^2_{\rho}+b^2_{\rho}\sigma_1^2+\frac{\sigma_1^2}{4\delta^2})\| \tilde z_\mathcal{T}\|^2+\varepsilon \delta^2\| \tilde w\|^2\\
       & +\varepsilon\frac{\|\dot {\bar w}\|^2}{4}+\varepsilon\frac{\| {\bar w}\|^2}{4}\\
        \varepsilon \Delta V'_3\leq& \varepsilon(\sigma^2_1 b^4_{\rho}+b^4_{\rho}+2\sigma_1b_{\rho}^2)\| \tilde w\|^2+\varepsilon \|\dot {\bar w}\|^2+\varepsilon\| {\bar w}\|^2,
                 \end{aligned}   \end{equation} 
where $\delta>0$, $b_\rho:=b\rho(\Delta)$ defined in \eqref{333} and  $b:=\| \bar E_\mathcal{T}^\top \bar E_\odot \|$ for all $(t,x)\in\mathbb{R}\times\mathbb{B}(\Delta)$.

    For the closed-loop system  \eqref{closesystem}, based on \eqref{353} and \eqref{n211},  $\dot V'(t,\xi)$ becomes
 \begin{equation}\begin{aligned}
     \label{351}
     \dot V'(t,\xi)	= & \dot V(t,\xi)+\Delta V'_1+\varepsilon \Delta V'_2+\varepsilon \Delta V'_3\\
       \leq &  -\alpha \| \tilde z_\mathcal{T}\|^2 -\beta\|\tilde w\|^2+(\delta^2-\sigma_1+\delta^2 \sigma_1^2)\| \tilde w\|^2\\
       &+\frac{\| {\bar w}\|^2}{4\delta^2}+\varepsilon (2\sigma_1  b^2_{\rho}+\sigma^2_1 b^4_{\rho}+b^4_{\rho}+\delta^2) \| \tilde w\|^2\\
      & +\varepsilon(b^2_{\rho}+b^2_{\rho}\sigma_1^2+\frac{\sigma_1^2}{4\delta^2})\| \tilde z_\mathcal{T}\|^2 +\frac{\|\dot {\bar w}\|^2}{4\delta^2}\\
      &+\varepsilon\frac{\|\dot {\bar w}\|^2}{4}+\varepsilon\frac{\| {\bar w}\|^2}{4}.\\
     \end{aligned}
 \end{equation}
Choosing $\beta_1:=\beta+\sigma_1-\delta^2 -\sigma_1^2\delta^2-\varepsilon (2\sigma_1  b^2_{\rho}+\sigma^2_1 b^4_{\rho}+b^4_{\rho}+\delta^2)>0$ and $\alpha ':=\alpha-\varepsilon(b^2_{\rho}+b^2_{\rho}\sigma_1^2+0.25\sigma_1^2/{\delta^2})>0$ yields
 \begin{equation}\begin{aligned}
     \label{352}
     \dot V'(t,\xi)	\leq &  -\alpha' \| \tilde z_\mathcal{T}\|^2 -\beta_1\|\tilde w\|^2+d_{\xi}\\
       \leq & -c_3 \|\xi\|^2+d_{\xi},
     \end{aligned}
 \end{equation}
 where $\beta$ is defined in \eqref{353},  $ c_3:=\min\{\alpha',\beta'\}$ and $d_{\xi}:=\sqrt{(1+4\delta^2\varepsilon)(w^2_d+w'^2_d)/{4\delta^2}}$.
 Since $\beta_1$ depends on $\mu$ and $T$ from Definition~\ref{def:dPE}, it is possible to choose $\beta_1 >0$. Parameter $\varepsilon$ is chosen to be sufficiently small and $\alpha_3=c_1'$ can be chosen to be sufficiently large so that $\alpha'>0 $. 
 Therefore, the solution $\xi$  of \eqref{closesystem} converges to  $\Omega_{\xi}:=\{ \xi: \|\xi\| \leq d_{\xi} \}$ with $d_{\xi}:=\sqrt{ {d}/{c_3} }$. The weight estimation errors converge to $\Omega_{\tilde w} :=\{ \tilde w: \|\tilde w\| \leq \sqrt{ d/{c_3} } \}$. 
 Furthermore, we obtain the bound of synchronization errors by  $z=R^{\top}z_{\mathcal{T}}$ and $z_{\mathcal{T}}=\tilde z_{\mathcal{T} }+\hat z_{\mathcal{T}}$. According to \eqref{237}, $\hat z_{\mathcal{T}}$ converges to $\Omega_{\hat z_{\mathcal{T}}}:=\{\hat z_{\mathcal{T}}: \|\hat z_{\mathcal{T}}\| \leq  d_{\hat z_{\mathcal{T}}}=\|\rho(\|d_{\xi}\|)\|/{c_2'}  \}$. The edge state $z$ thus converges to $\Omega_z:=\{ z: \|z\| \leq \|R^{\top}\|\|d_z\| \}$ where $d_z=[d_{\xi}\quad  d_{\hat z_{\mathcal{T}}}]$.
 \end{proof}

\section{Simulation}\label{sec:simulations}
We consider a network \eqref{108} with 6 agents with a time-varying communication topology with  $\bar w(t)$ in \eqref{99} as
\begin{equation*}\begin{aligned}
    \bar w(t)=&[0.7+0.02\sin(0.02t),0.8+0.1\cos(0.01t),0.6+\\&0.02\sin(0.5\pi t),0.25,0.4,
    0.02\cos(0.05\pi t)+0.45,\\&\emph{0}_{1\times15},0.05\cos(0.01\pi t)+0.3,0.6,0.2,\emph{0}_{1\times 5},0.5]^\top\\\end{aligned}
\end{equation*}
where $\emph{0}_{1\times N}$ denotes $N$-dimensional zero row vector.
 Here, we simulate the network \eqref{108} with $f_i(x_i)=x_i$ and $c=1$.  Use control input \eqref{166} and weight updating law \eqref{170_n} and design the auxiliary system \eqref{313}.  The control gains are chosen as $c_1=2, c_2=1.3, \sigma_1=0.001$. Choose a u$\delta$-PE function from \eqref{313} referring to \cite{restrepo2023simultaneous}  as
 \begin{align*}
    \psi(t,\tilde z_\mathcal{T})= &(\bar E^\top)^+\tanh(\kappa\bar E_{\mathcal{T}} \tilde z_{\mathcal{T}})p(t),\\
p(t)=&5\sin(0.5\pi t)+4\cos(2\pi t)-6\sin(8\pi t)+\sin(\pi t)\\
  &- 4\cos(10\pi t)+2\cos(6\pi t)+3\sin(3\pi t).\end{align*}

The simulation results are presented in Figs~\ref{fig:weight}-\ref{fig:synchronized_zk}. Fig~\ref{fig:weight} and \ref{fig:weight_errors} represent the evolution of the estimated weights and the errors between the estimated weight and the time-varying weight.  Fig~\ref{fig:zstate} shows the evolution of synchronization errors $z$. Fig~\ref{fig:synchronized_zk} displays the evolution of state   $\tilde z_{\mathcal{T}}$. 
As expected from Proposition~\ref{propUSPAS}, the estimation errors, $\tilde z_{\mathcal{T}}$ and synchronization errors $z$ are bounded from Fig.~\ref{fig:weight_errors}, \ref{fig:zstate} and \ref{fig:synchronized_zk} under the time-varying topology. From Figs.~\ref{fig:weight} and \ref{fig:weight_errors}, the real time-varying weights are in the line segments whose centres are the predicted weight in Fig.~\ref{fig:weight} and whose radii are the weight estimation errors. Another observation is that the bounds of the synchronization errors $z$ are bigger than the bounds of the estimation weight errors $\tilde w$, which responds to the analysis in the proof part of Proposition~\ref{propUSPAS}. We also tried different values of  $c_2$, and we found that increasing the value of $c_2$ could get lower synchronization errors while increasing the bound of the weight estimation errors, which correspond to the form \eqref{313} of the auxiliary system $\hat z$. Hence, keeping a certain level of excitation for $\tilde z_{\mathcal{T}}$ is beneficial to estimating the time-varying weights, while it deteriorates the synchronization.

\begin{figure*}[!h]
	\centering
 \begin{minipage}[b]{0.2\textwidth}
 \centering
\includegraphics[scale=0.3]{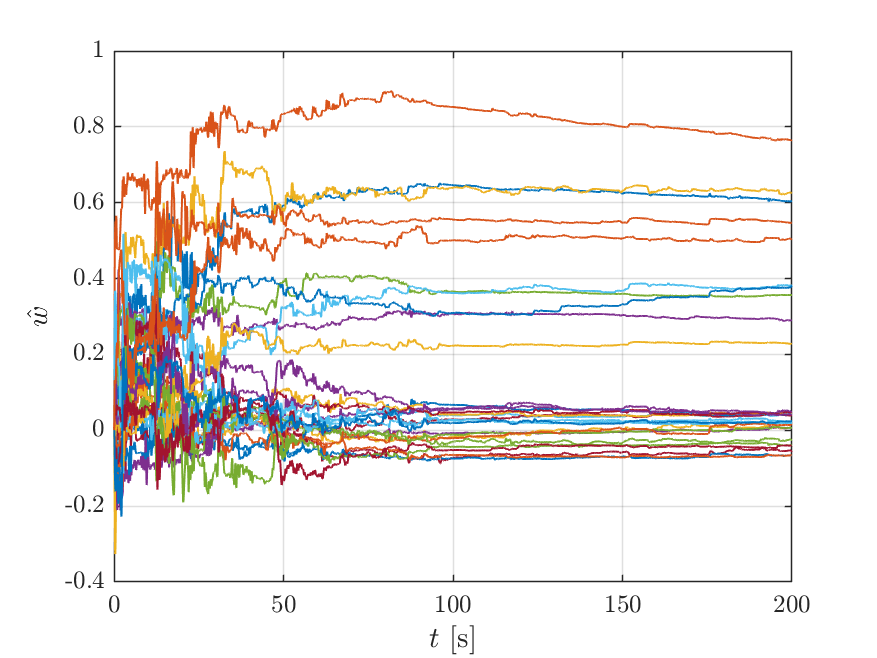}
		\caption{Estimated weight of time-varying topology}
		\label{fig:weight}
	\end{minipage}
	\hfill
 \begin{minipage}[b]{0.2\textwidth}
		\centering
		\includegraphics[scale=0.3]{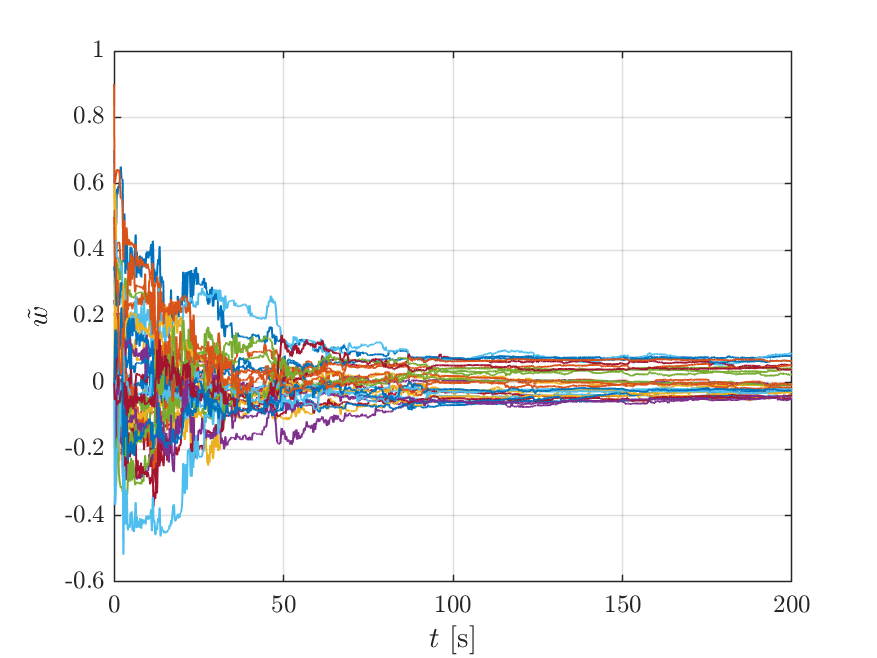}
		\caption{Estimation errors of  time-varying weights.}
		\label{fig:weight_errors}
	\end{minipage}
	\hfill
  \begin{minipage}[b]{0.2\textwidth}
	\includegraphics[scale=0.3]{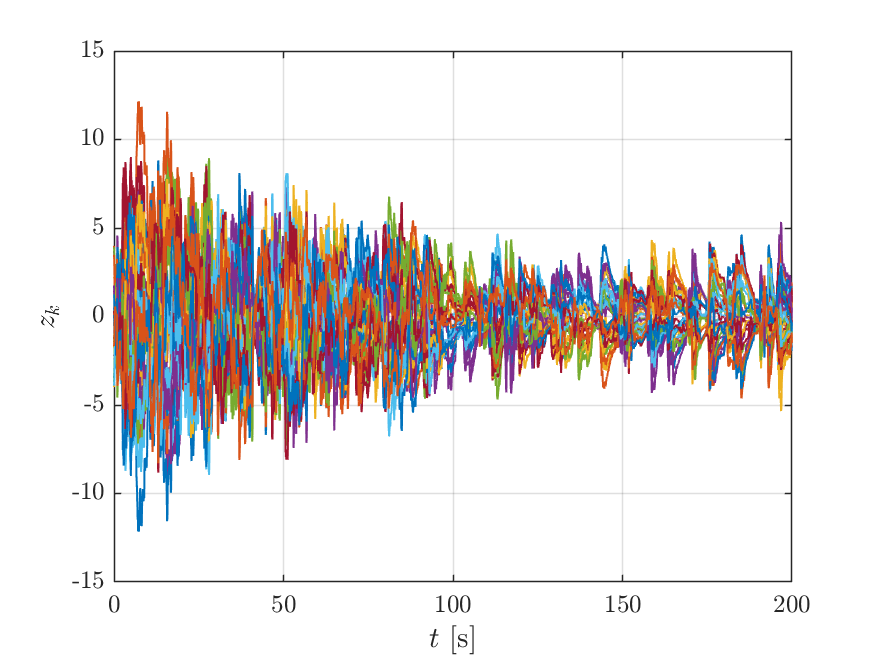}
	\caption{Evolution of synchronization errors $z$}
	\label{fig:zstate}
	\end{minipage}
	\hfill
 \begin{minipage}[b]{0.2\textwidth}
	\centering
	\includegraphics[scale=0.3]{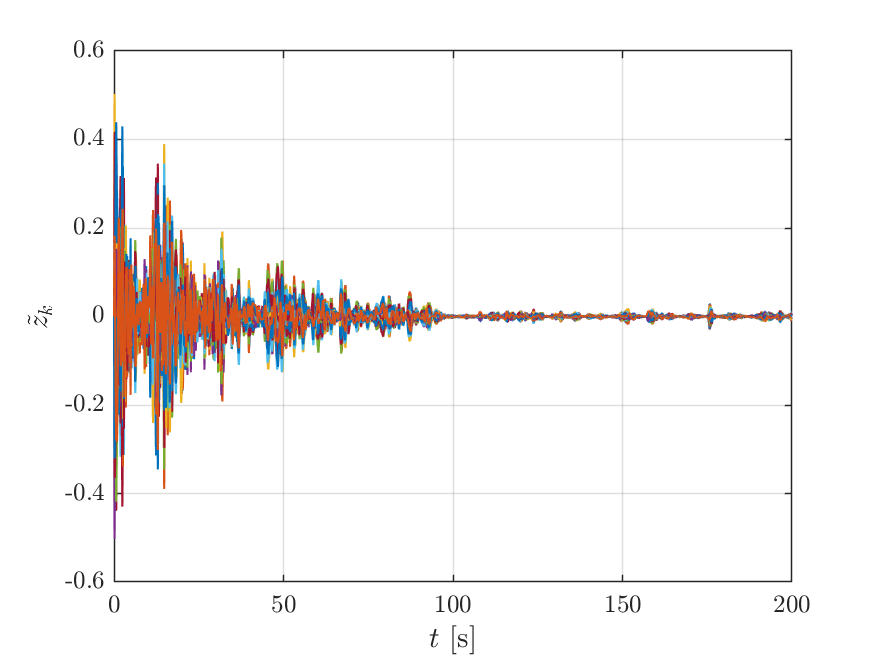}
	\caption{Evolution of state $\tilde z_{\mathcal{T}}$}
	\label{fig:synchronized_zk}
\end{minipage}
\end{figure*}

\section{Conclusions}\label{sec:conc}

In this paper, we introduce an adaptive control-based approach for simultaneous estimation of time-varying topology and synchronization of a complex dynamical network. We design an adaptive-control-based scheme to stimulate the system to ensure the boundedness of topology estimation errors. This is achieved through the development of an auxiliary system characterized by either persistent excitement or uniform $\delta$ persistent excitement.  The first auxiliary system which is PE, enables us to bound the edge weight estimation errors. The latter one which is uniformly-$\delta$ persistently exciting gives the boundedness of both weight estimation errors and synchronization errors, provided the weights and their derivatives are bounded. In terms of further work, we aim to enhance the topology estimation performance while considering control tasks under time-varying topology.

\bibliographystyle{IEEEtran}
\bibliography{./ref}

\begin{appendices}\label{sec:apenn}

\section{On $\delta$-persistency of excitation}\label{apx:dPE}

\begin{definition}[Persistency of excitation ]\label{Definition:PE}
	 A function $\phi:\mathbb{R}_{\geq 0}\to\mathbb{R}^{n\times m}$ is said to be persistently exciting, if there exist positive $T$ and $\mu$ such that  for any unit vector $v\in R^m$,
	\begin{equation}\label{eq:PE}
	\int_{t}^{t+T}  \| \phi(\tau)v\|^2 d\tau \geq\mu, \forall t\geq0.
	\end{equation}
\end{definition}

Partition $x\in\mathbb{R}^n$ as $x:=\left[x_1^\top\;\; x_2^\top\right]^\top$ where $x_1\in\mathbb{R}^{n_1}$ and $x_2\in\mathbb{R}^{n_2}$. 
Define the set $\mathcal{D}_1:=\left(\mathbb{R}^{n_1}\backslash\{0\}\right)\times\mathbb{R}^{n_2}$ and the function $\phi:\mathbb{R}\times\mathbb{R}^n\to\mathbb{R}^m$ where $t\mapsto\phi(t,x)$ is locally integrable. The following defines uniformly $\delta$-persistently exciting from \cite[Lemma~1]{loria2005matrosov}.

\begin{definition}\label{def:dPE}
	[Uniformly $\delta$-persistency of excitation]
	If $x\mapsto\phi(t,x)$ is continuous uniformly in $t$, then $\phi(\cdot,\cdot)$ is uniformly $\delta$-persistently exciting (u$\delta$-PE) with respect to $x_1$ if and only if for each $x\in\mathcal{D}_1$ there exist positive $T$ and $\mu$ such that for any unit vector $v\in R^m$
	\begin{equation}\label{eq:dPE}
	\int_{t}^{t+T}\|\phi(\tau,x)v\|^2 d\tau\geq\mu, \forall t\geq0.
	\end{equation}
\end{definition}

The next Lemma establishes that when a strictly proper stable filter is subject to a bounded disturbance, and driven by a u$\delta$-PE input, its output retains the property of being u$\delta$-PE. The lemma, originally introduced in \cite{panteley2001relaxed}, didn't account for the presence of a bounded disturbance. 

\begin{lemma}[Filtration property]\label{lemma1}
	Let $\phi:\mathbb{R}_{\geq 0}\times\mathbb{R}^{n}\to\mathbb{R}^{p\times q}$ and consider the system 
	\begin{equation} \label{app_sys}
	\left[\begin{matrix}
	\dot x\\\dot\omega
	\end{matrix}\right]=\left[\begin{matrix}f(t,x,\omega)\\
	f_1(t,\omega)+f_2(t,x)\omega+\phi(t,x)\end{matrix}\right]
	\end{equation}
	where $f_1:\mathbb{R}_{\geq 0}\times\mathbb{R}^{n}\to\mathbb{R}^{p\times q}$ is Lipschitz in $\omega$ uniformly in $t$ and measurable in $t$ and satisfies $\|f_1(\cdot)\|\leq l \|\omega\|$ for all $t$; $f_2:\mathbb{R}_{\geq 0}\times\mathbb{R}^{n}\to\mathbb{R}^{p\times p}$ is locally Lipschitz in $x$ uniformly in $t$ and measurable in $t$. Assume that $\phi(t,x)$ is u$\delta$-PE with respect to $x$. Assume that $\phi$ is locally Lipschitz and there exists a non-decreasing function $\alpha:\mathbb{R}_{\geq 0}\to\mathbb{R}_{\geq 0}$, such that, for all $(t,x)\in\mathbb{R}_{\geq 0}\times\mathbb{R}^n$:
	\begin{equation}\label{n751}
	\max\left\{\|\phi(\cdot)\|,\|f_2(\cdot)\|,\left\|\frac{\partial \phi(\cdot)}{\partial t}\right\|,\left\|\frac{\partial \phi(\cdot)}{\partial x}\right\|\right\}\leq\alpha(\|x\|).
	\end{equation}
Assume that $f(\cdot)$ satisfies that 
 \begin{equation}\label{n752}
	\max\left\{\|f(\cdot)\|\right\}\leq\alpha(\|x\|)+k,
	\end{equation}
 where  $k$ is a positive constant. Denote $w=(w_1,w_2,\cdots, w_p)^\top$ and $w_i^\top \in R^{q}$ with $i=1,2,\cdots, q$.  
If all solutions $x_\phi(t)$, defined as $x_\phi:=\left[x^\top\;\;\omega_1 \;\;\omega_2 \;\cdots\; \omega_p  \right]^\top$, satisfy
	\begin{equation}\label{n755}
	\|x_\phi(t)\|\leq r \quad \forall t\geq t_0,
	\end{equation}
	for a positive constant $r$, then $\omega$ is uniformly $\delta$-persistently exciting with respect to $x$.
\end{lemma}
\begin{proof}
Denote $v\in R^p$ as a unit vector.
	Defining $\rho:=-v^\top \phi\omega^\top v$,  we have\begin{equation}\label{n748}\begin{aligned}[b]
	\dot\rho=&-\|\phi^\top v\|^2-v^\top \phi f_1^\top v-v^\top \left[f_2\phi+\frac{\partial\phi}{\partial t}+\frac{\partial\phi}{\partial x}f\right]\omega^\top v\\
	\leq&-\|\phi^\top v\|^2+\|\omega ^\top v\|\left[2\alpha^2(r)+(l+k+1)\alpha(r)\right]\|v\|\\
 =&-\|\phi^\top v\|^2+c(r)\|\omega^\top v\|,
	\end{aligned}\end{equation}
where $c(r):=2\alpha^2(r)+(l+k+1)\alpha(r)$.
 Integrating both sides of 
  \eqref{n748}
 from $t$ to $t+T_f$ and then reversing the inequality sign, we derive that
	\begin{equation}\label{n767}\begin{aligned}[b]
	v^\top \phi(t,x) \omega(t)^\top&v- v^\top\phi(t+T_f,x) \omega(t+T_f)^\top v\\
 \geq\int_{t}^{t+T_f}\|\phi(\tau,x)^\top& v\|^2d\tau - \int_{t}^{t+T_f}c(r)\|\omega(\tau)^\top v\|d\tau.
	\end{aligned}\end{equation}
	By applying the bounds in \eqref{n751}, \eqref{n752} and \eqref{n755} to the left-hand side of inequality \eqref{n767}, we have
	\begin{equation*}
	2\alpha(r)r\geq\int_{t}^{t+T_f}\|\phi(\tau,x)^\top v\|^2d\tau - \int_{t}^{t+T_f}c(r)\|\omega(\tau)^\top v\|d\tau.
	\end{equation*}
	Let $T_f:=k'T$. Since $\phi(t,x)$ is u$\delta$-PE from \eqref{eq:dPE}, there exists $\mu$ such that
	\begin{equation*}
	\int_{t}^{t+k'T}\|\phi(\tau,x)^\top v\|^2d\tau\geq k'\mu.
	\end{equation*}
	Thus, we obtain
\begin{equation*}\begin{aligned}[b]
\int_{t}^{t+k'T}\|\omega(\tau)^\top v\|^2 d\tau\geq \frac{\left(k'\mu-2\alpha(r)r\right)^2}{c(r)^2}=:\mu_r.
	\end{aligned}\end{equation*}
	Choosing $k'$ large enough so that $\mu_r>0$, $\omega(t)$ is u$\delta$-PE with respect to $x$.
\end{proof}

\end{appendices}

\end{document}